\documentclass[runningheads]{llncs}
\usepackage{packages}
\usepackage{macros}

\begin{document}
\title{A Rewriting Logic Semantics and Statistical Analysis for
  Probabilistic Event-B}
\author{
  Carlos Olarte\inst{1} \and
  Camilo Rocha\inst{2} \and
  Daniel Osorio\inst{2} 
}
\authorrunning{C. Olarte, C. Rocha and D. Osorio}

\institute{LIPN, Universit\'{e} Sorbonne Paris Nord, Villetaneuse,
  France \and Pontificia Universidad Javeriana, Cali, Colombia.}
\maketitle
\begin{abstract}
  Probabilistic specifications are fast gaining ground as a tool for
  statistical modeling of probabilistic systems. One of the main goals
  of formal methods in this domain is to ensure that specific behavior
  is present or absent in the system, up to a certain confidence
  threshold, regardless of the way it operates amid uncertain
  information. This paper presents a rewriting logic semantics for a
  probabilistic extension of Event-B, a proof-based formal method for
  discrete systems modeling.
  The proposed semantics adequately captures the three sources of
  probabilistic behavior, namely, probabilistic assignments,
  parameters, and concurrency.  Hence, simulation and probabilistic
  temporal verification become automatically available for
  probabilistic Event-B models.
  The approach takes as input a probabilistic Event-B specification,
  and outputs a probabilistic rewrite theory that is fully executable
  in PMaude and can be statistically tested against quantitative
  metrics. The approach is illustrated with examples in the paper.
\end{abstract}

\section{Introduction}\label{sec:intro}
For many systems, there is an obvious need for using specialized
formal methods in the spirit of formalisms, inference systems, and
simulation techniques for selected tasks.
When properly combined, formal methods have a great potential to
become more useful in practice, and scale up because of modularization
and specialization of needs.
In the realm of probabilistic systems, where a vast number of
randomized algorithms and protocols fall, both inference- and
algorithmic-based analysis techniques are needed to answer the key
question of whether such systems are correct. Probabilistic systems
must behave properly, up to a confidence threshold, regardless of the
way they operate amid uncertain information.
A challenging task is to find combinations of formalisms, inference
systems, and simulation techniques for verifying such systems.


This paper focuses on probabilistic simulation and statistical
analysis for a probabilistic extension of Event-B
\cite{DBLP:books/daglib/0024570}, a formal method for system-level
modeling and analysis. Event-B uses set theory as a modeling notation,
refinement to represent systems at different abstraction levels, and
mathematical proofs to verify consistency between refinement levels
(and other proof obligations).
The probabilistic extension of Event-B proposed
by~\cite{aouadhi:hal-01255753} is based on three mechanisms.
First, interaction with the external environment is governed by a
probabilistic choice; that is, external inputs are chosen uniformly
from finite sets representing the potential values for the variables
under the control of the external environment.
Second, the extension includes probabilistic assignment to variables
(representing the state of the system) by uniform choice from finite
sets. This is useful to handle the uncertainty associated to
``external'' data.
Third, all concurrent transitions in a system are probabilistic. That
is, in the presence of the two first mechanisms, concurrency yields a
probabilistic transition system where each possible transition from a
state is weighted by a probability measure.
The authors of~\cite{DBLP:journals/sosym/AouadhiDL19} have proposed an
inference system for reasoning within this probabilistic extension of
Event-B.
The purpose of the presented paper is to complement the work in
\cite{DBLP:journals/sosym/AouadhiDL19} by enabling algorithmic
simulation and statistical model checking.


This paper develops a rewriting logic semantics for the
above-mentioned probabilistic extension of Event-B. It takes as input
a probabilistic Event-B model and outputs a probabilistic rewrite
theory~\cite{agha-pmaude-2006}. The mapping implementing this
translation is explained in detail, including its soundness and
completeness properties for simulation in relation to the given
model. The resulting probabilistic specification is executable in
Maude~\cite{clavel-maudebook-2007} and it is amenable to
simulation-based verification such as, e.g., statistical model
checking with help of the PVeStA model checker
tool~\cite{alturki-pvesta-2011}. The translation has been fully
automated and it supports a wide range of Event-B operators.


The approach presented here can be seen as a complement to the
proof-theoretic techniques developed in
\cite{DBLP:journals/sosym/AouadhiDL19}. It allows system designers to
experiment with the system via simulation and automatically verify
system's properties via (stochastic) model checking, thus gaining more
confidence on the model before embarking on a proving task. It is
known that Event-B offers both proof-theoretic methods, and simulation
and model checking tools (via, e.g., the ProB animator and model
checker~in the Rodin platform
\cite{DBLP:books/daglib/0024570,DBLP:journals/sttt/LeuschelB08}) for
system modeling. These latter features do not exist, to the best of
the authors' knowledge, for probabilistic Event-B models.
Ultimately, the developments presented here open new opportunities for
incarnating the more ambitious long-term ideal of combining
formalisms, inference systems, and simulation techniques for verifying
probabilistic systems.

\noindent{\bf Outline.}
Section~\ref{sec:preventb} presents an overview of the probabilistic
extension of Event-B.
Section~\ref{sec:enc} develops the mapping from Event-B models to
probabilistic rewrite theories and studies its main formal properties.
Section~\ref{sec:cases} showcases an example of the transformation and
its statistical analysis. The web page of the companion tool
\cite{tool.website} of this paper presents other case studies 
and provides some  experimental results. 
Finally, Section~\ref{sec:conc}  concludes the paper.

\section{Probabilistic Event-B in a Nutshell}\label{sec:preventb}

This section presents an overview of the probabilistic extension of
Event-B proposed by~\cite{DBLP:journals/sosym/AouadhiDL19}. As it is
the case for the non-probabilistic case, a probabilistic Event-B model
consists of a \emph{context} (\S \ref{sec:context}) and a
\emph{machine} (\S \ref{sec:machine}).

\vspace{-0.3cm}
\subsection{Contexts}\label{sec:context}
A \textit{context} in Event-B describes the constant part of the
system, including the definition of deferred sets and constants. More
precisely, a context is a triple $\mathcal{C} = \langle SET_l \cup
SET_n, CTE\rangle$ where $SET_l$ is a set of set declarations of the
form $S : \{id_1,...,id_m\}$, $SET_n$ is a collection of set
declarations of the form $S : n$ (where $n$ is a natural number), and
{$CTE$} is a set of constant definitions of the form $c ~:~ type ~:=~
v$.

\noindent\textbf{Deferred sets}.  A typical context can define any
number of deferred sets, and axioms indicating either that any such a
set is finite (but its cardinality is unknown) or has a fixed
cardinality $n$.  The former construction is not considered here,
while the latter is supported via the $SET_n$ component. For instance,
the declaration \code{ENUM:3} generates three constant symbols (e.g.,
\code{ENUM$i$}, for $1 \leq i \leq 3$) that inhabit the set
\code{ENUM}.  Moreover, the deferred set declaration
\code{STATE:\{open,close\}} defines the set \code{STATE} inhabited
exactly by the two (distinct) constant symbols \code{open} and
\code{close}.

\noindent\textbf{Types, values, and constants.} The syntax
\code{DEFAULT:STATE :=open} defines the constant \code{DEFAULT} of
type \code{STATE} with value \code{open}. Basic types include deferred
sets, Boolean values, and integer numbers. Types can be also built via
Cartesian products (\code{A * B}) and powersets (\code{POW(A)}). From
those constructions, relations and (partial, total, injective, etc.)
functions can be defined as expected. For instance, the set-theoretic
language of Event-B reduces $r \in A \rel B$ ($r$ is a relation from
$A$ to $B$) into $r \subseteq A \code{*} B$ and finally into $r \in
\code{POW}(A \code{*} B)$.  For succinctness, $A$ and $B$ in $A
\code{*} B$ are restricted here to be basic types.

According to the types above, the values for constants can be
elements of a deferred set, numbers, Booleans, pairs ($a \mapsto b$), set
of values as in $\{a,b,c\}$, or integer intervals as in $1 .. 42$. The complete
grammar can be found in \cite{tool.website}. 

\noindent\begin{wrapfigure}{l}{0.33\textwidth}
\begin{minipage}{.31\textwidth}
\vspace{-0.8cm}
\begin{EventB}
CONTEXT GEAR_CTX   
 SETS  
   SUD:{ up, down }   
   SER:{ extended, retracted }   
   SOC:{ open, close }
 CONSTANTS  
   FCMD : Nat := 9    
END 
\end{EventB} 
\vspace{-1cm} 
\end{minipage}
\end{wrapfigure} 

As a running example, consider the controller for a landing gear
system modeled in Event-B in~\cite{DBLP:journals/sttt/BoniolWAS17} and
probabilized in~\cite{aouadhi:hal-01255753}.  When landing, the
following sequence of actions occur: the doors of the system are
opened, the landing gears are extended, and then the doors are
closed. Similarly, after taking off, the doors are opened, the gears
retracted, and the doors are closed. The pilot may initiate and
interrupt these sequences with a handle that can be in two positions:
\emph{up} (executing the retracting sequence) and \emph{down}
(extending sequence). The context \code{GEAR_CTX} defines the three
needed deferred sets and a constant of type $\textit{Nat}$ that will
be used in the next section.

\vspace{-0.2cm}

\subsection{Machines}\label{sec:machine}
A \textit{machine} in Event-B specifies a set of variables, defining
the state of the system, and the system's actions, called
\emph{events}.  More precisely, a \textit{machine} is a structure of
the form $\mathcal{M} = \langle\mathcal{C}, \vec{x}, I, \mathcal{E},
init \rangle$ where $\mathcal{C}$ is a context, $\vec{x}$ is a set of
variables typed by the invariant $I$, $\mathcal{E}$ is a set of
probabilistic events, and $\textit{init}$ is the initialization event.

\noindent\textbf{State of the machine. } In an Event-B specification, 
the section \code{SEES} of a machine
determines the context $\mathcal{C}$ accessible by the model.  For
simplicity, it is assumed here that each machine \emph{sees} exactly
one context.  The section \code{VARIABLES} contains a list of
identifiers.  The types (or domains) for each variable (e.g.,
\code{currentState:STATE}) are defined in the section
\code{INVARIANTS}. Besides typing information, Event-B models include
also invariant properties that must be preserved along the system's
transitions. The translation in Section \ref{sec:enc} requires only
the types for the variables and hence, for the moment, only invariants
of the form $x : A$ (the type of $x$ is $A$) are considered.  Section
\ref{sec:cases} shows how the infrastructure presented here can be
used to verify some other properties.

\noindent\textbf{Events.} The machine \emph{initialization} is
deterministic and the 
section \code{INITIALISATION} assigns values (of the appropriate type)
to all the variables (e.g., \code{currentState := open}).  Such values
are built from elements of deferred sets, constants in the context,
standard arithmetic and Boolean operations, and expressions on sets
and relations (see \cite[Chapter 6]{DBLP:books/daglib/0024570} for the complete 
set-theoretic language).

\noindent\begin{wrapfigure}{l}{0.13\textwidth}
~\begin{minipage}{.13\textwidth}
\vspace{-1.0cm}
\begin{EventB}
EVENT ID
  WEIGHT W
  ANY P
  WHERE G
  THEN A
END
\end{EventB}
\end{minipage}
\vspace{-1.0cm}
\end{wrapfigure}
The system's actions are defined by \emph{events}, which are composed
of \emph{guards} and \emph{actions}.  In the case of probabilistic
models, each event is also assigned a \emph{weight}.  At a given
state, the numerical expression \code{W} determines the weight of the
event with identifier \code{ID}. Such an expression is built from
constants, the machine's variables, and arithmetic operations. As
explained later in the semantics, events with higher weights are more
likely to be chosen for execution.

The variables declared in the optional section \code{ANY} are called
\emph{parameters}; they represent an interaction with an external
environment, which is not under the control of the system. Hence, the
behavior of the event is not fixed but it may have different outcomes
depending on the values chosen for those parameters. The expressions
$y :\in S$ states that the parameter $y$ may take values from the
non-empty and finite set $S$.  The set $S$ can be a an arbitrary
expression returning a set, including for instance: set comprehension
($\{x. S \mid P(x)\}$, elements of $S$ that satisfy the predicate $P$
and $\{x.S \mid F(x)\}$ defining the set $\{F(x) \mid x \in S\}$); the
Cartesian product; domain of a relation; range restriction ($r \ranres
S = \{x\mapsto y \in r \mid y\in S\}$); domain subtraction ($S\domsub
r = \{x\mapsto y \in r \mid x \notin S\}$); overriding ($ r \ovl s = s
\cup (\dom{(s)} \domsub r)$; etc.  The complete list of expressions
involving set/relations (union, membership, range restriction,
cardinality, etc.) currently supported by our tool can be found
at \cite{tool.website}.
Semantically, in $y
:\in S$, a value from $S$ is chosen with a uniform probability
distribution and assigned to the parameter $y$.

The Boolean expression \code{G} in the \code{WHERE} section determines
whether the event can be fired at a given state or not. The
\emph{action} of the event determines the new state and it is
specified by a list of simultaneous assignments in the \texttt{THEN}
section.  It is assumed that the right hand side (RHS) of an
assignment is an expression of the appropriate type including
constants and variables (considering the values before the
assignment), as well as the event's parameters. When the RHS is a list
of expressions, a uniform probabilistic distribution is assumed for
them. Moreover, it is also possible to use an enumerated probabilistic
assignment as in $x := \{\texttt{open}@0.7, \texttt{close}@0.3\}$,
specifying that $x$ may take the value \texttt{open} (resp.,
\texttt{close}) with probability $0.7$ (resp., $0.3$).

\noindent\begin{wrapfigure}{l}{0.34\textwidth}
\vspace{-1cm}
\begin{minipage}{.33\textwidth}
\begin{EventB}
MACHINE GEAR SEES GEAR_CTX
 VARIABLES handle gear door cmd
 INVARIANTS
   handle : SUD
   gear : SER
   door   : SOC
   cmd  : Nat 
 INITIALISATION
   handle := up
   cmd    := 0
   gear   := retracted
   door   := closed
\end{EventB}
\end{minipage}
\vspace{-1cm}
\end{wrapfigure}
In the running example, the pilot may use the handle to initiate and
interrupt the extending and retracting sequences. The machine
\code{GEAR} uses three variables to observe the current state of the
handle, the gear, and the doors.  The variable $\textit{cmd}$ controls
the number of times the pilot has initiated the sequence. \\

The event \code{extend} in Figure \ref{fig:events} models the
extension of the gear when the handle is down and the doors are
opened. As a model of failures, there is a $10\%$ of risk that the
gears do not react correctly to their command.  The event
\code{retract} can be explained similarly.  The state of the doors are
controlled by the events \code{open} and \code{close}. For instance,
if the doors are closed, the gear retracted, and the current command
is \emph{extend} (\code{handle=down}), the doors are opened with
probability 0.9.

The interface of the system with the pilot is modeled with the event
$pcmd$ below.  The requirements of the system impose that: the pilot
cannot command the handle

\begin{wrapfigure}{l}{0.22\textwidth}
\vspace{-1cm}
\begin{minipage}{.22\textwidth}
\begin{EventB}
EVENT pcmd 
  WEIGHT 
    FCMD - cmd
  ANY    
    cc :$\in$ {up, down}
  WHERE  
    cmd <= FCMD
THEN   
    handle := cc
    cmd    := cmd + 1
END
\end{EventB}
\end{minipage}
\vspace{-1cm}
\end{wrapfigure}
\noindent more than a fixed number of times before one of the sequences begins;
and consecutive uses of the handle must decrease the priority of using it again.
Since the pilot may interrupt an already started sequence, this event modifies
the state of the handle with equal probability to \code{up} and \code{down}.
Due to the weight of the
event, such changes are allowed only up to \code{FCMD} times (the constant
defined in the context \code{GEAR_CTX}).

\begin{figure}[t]
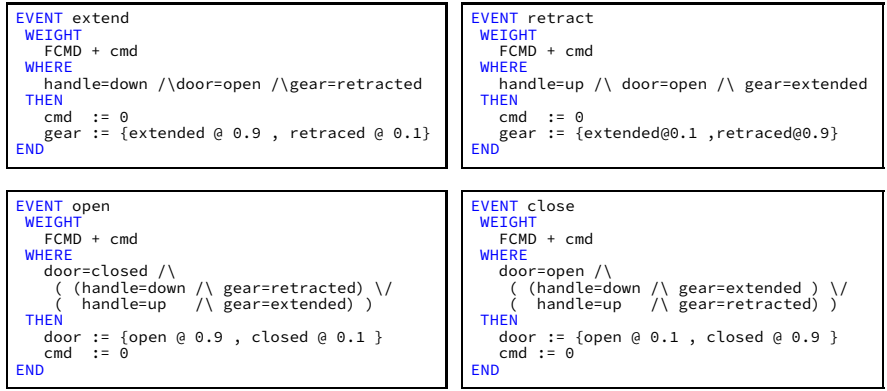

\noindent
\begin{minipage}{.46\textwidth}
\begin{EventB}
EVENT extend 
 WEIGHT 
   FCMD + cmd
 WHERE 
   handle=down /\door=open /\gear=retracted
 THEN 
   cmd  := 0
   gear := {extended @ 0.9 , retraced @ 0.1}
END
\end{EventB}
\end{minipage}
~~~
\begin{minipage}{.44\textwidth}
\begin{EventB}
EVENT retract 
 WEIGHT 
   FCMD + cmd
 WHERE 
   handle=up /\ door=open /\ gear=extended
 THEN 
   cmd  := 0
   gear := {extended@0.1 ,retraced@0.9}
END
\end{EventB}
\end{minipage}

\vspace{-0.2cm}
\noindent
\begin{minipage}{.46\textwidth}
\begin{EventB}
EVENT open 
 WEIGHT 
   FCMD + cmd
 WHERE 
   door=closed /\
    ( (handle=down /\ gear=retracted) \/ 
    (  handle=up   /\ gear=extended) ) 
 THEN 
   door := {open @ 0.9 , closed @ 0.1 }
   cmd  := 0
END
\end{EventB}
\end{minipage}
~~~
\begin{minipage}{.44\textwidth}
\begin{EventB}
EVENT close
 WEIGHT 
   FCMD + cmd
 WHERE 
   door=open /\
    ( (handle=down /\ gear=extended ) \/
    (  handle=up   /\ gear=retracted) ) 
 THEN 
   door := {open @ 0.1 , closed @ 0.9 }
   cmd := 0   
END
\end{EventB}
\end{minipage}
\caption{Events controlling the gear and the doors\label{fig:events}}
\vspace{-0.6cm}
\end{figure}

\vspace{-0.2cm}
\subsection{Probabilistic Semantics}\label{sec:prob-sem}
\vspace{-0.2cm}

The state of a machine is a valuation $\stateSym$ mapping variables
and constants to values of the appropriate type.  For a variable or
constant $x$, $\stateSym(x)$ denotes the value of $x$ in $\stateSym$.
Given an expression $E(\vec{x})$ that may depend on the (list of)
variables $\vec{x}$, the evaluation of $E$ in the context $\stateSym$
is denoted as $E(\vec{x})[\stateSym]$.  Moreover, for an event $e \in
\mathcal{E}$, $\evar(e)$ is the set of variables that appear on the
left hand side (LHS) of the assignments in the action of the event.
In each event, a variable can appear only once as LHS in the list of
assignments.

An event $e = \langle w(\vec{x}),\vec{y},G(\vec{x},\vec{y}), act
\rangle \in \mathcal{E}$ is \emph{enabled} at state $\stateSym$ if
$w(\vec{x})[\stateSym] > 0$ and there exists a valuation $\sigma$ for
the parameters $\vec{y}$ making the guard true, i.e.,
$G(\vec{x},\vec{y})[\stateSym]\sigma = {true}$. Then, the action $act$
can be performed.  Given a state $\stateSym$, $\eevt(\stateSym)$
denotes the set of enabled events at $\stateSym$.  The notation
$w_{e}$ is used to denote the weight expression of the event $e$;
similarly for the other components of the event.

The semantics of a machine $\mathcal{M}$ is a probabilistic labeled transition
system (PLTS) $\langle S, s_0, Acts, T\rangle $ where $S$ is the set of states,
$s_0$ is the valuation obtained after executing the initialization, $Acts$ is
the set of labels of the events, and $T : S \times Acts \times S \to [0,1]$ is
the transition probability function defined as follows: 

\noindent\resizebox{.85\textwidth}{!}{
$
(\stateSym, e, \stateSym') = 
\left\{
\begin{array}{ll}
0 & \mbox{ if } \stateSym \not\to \stateSym'
\\
\underbrace{
\frac{w_{e}[\stateSym]}{\sum\limits_{e\in \mathcal{E}} w_e[\stateSym]}
}_{weight}
\times
\sum\limits_{\sigma\in \mathcal{T}(\stateSym,e)}
\left(
\underbrace{
\frac{1}{|\mathcal{T}(\stateSym,e)|}}
_{parameters}
\times
\underbrace{
\prod\limits_{x\in \evar(e)}
~
\sum\limits_{E\in Val(x,\stateSym,\stateSym',\sigma,e)}P(E)
}_{assignments}
\right)
& \mbox{ otherwise}  
\end{array}
\right.
$
}

\noindent where:
 $\stateSym \not\to \stateSym'$ means that either
   $\eevt(\stateSym)=\emptyset$ or $x[\stateSym] \neq x[\stateSym']$
   for a variable $x$ not in $\evar(e)$ (i.e., $\stateSym$ and
   $\stateSym'$ differ on a variable not modifiable by the event);
   $\mathcal{T}(\stateSym, e)$ is the set of valuations for the
   parameters of the event $e$ that make its guard true, i.e.,
   $\sigma\in \mathcal{T}(\stateSym, e)$ iff
   $G_{e}(\vec{x},\vec{y})[\stateSym]\sigma = \true$; 
   $Val(x,\stateSym,\stateSym',\sigma,e)$ is the set of expressions
   on the RHS of the assignment that assign to $x$ its value in
   $\stateSym'$, i.e.,  $E \in
   Val(x,\stateSym,\stateSym',\sigma,e)$ iff $E[\stateSym]\sigma =
   x[\stateSym']$;  and $P(E)$ is the probability of choosing
   the expression $E$ among all the other expressions.

The aim of the first term (\emph{weight}) is to normalize the weights
of the enabled events.  Hence, the larger $w_{e}[\stateSym]$ is, the
higher the probability of choosing $e$ for execution when
enabled. Since the expression $w_e$  depends on the state of the
variables, such probability may change during the execution of the
system. The second term (\emph{parameters}) counts the
number of possible assignments for the parameters that satisfy the
guard of the event at state $\stateSym$.  Recall that the domain for
the parameters must be a non-empty finite set.  The third term
(\emph{assignments}) considers all possible expressions that can be
generated due to assignments with a list of expressions (uniform
distribution) or enumerated probabilistic assignments.  
In the following, $\stateSym\rede{(p,e)} \stateSym'$ denotes 
that $T(\stateSym,e,\stateSym')=p$ and $p>0$.

The work in \cite{DBLP:journals/sosym/AouadhiDL19} establishes
conditions for the PLTS generated by an Event-B specification to be a
Discrete Time Markov Chain (DTMC). For that, some proof obligations
must be discarded: (1) the weight of the events are natural numbers;
and (2) in a probabilistic assignment $x:=\{E_1@p_1,...,E_n@p_n\}$,
$0<p_i\leq 1$ and $\sum p_i = 1$.  In that case, it is possible to
show that for all states $\stateSym$, $\sum \{p \mid \stateSym
\rede{(p,e)}\stateSym'\} = 1$. These conditions are assumed to be true
in the rest of the paper for any machine $\mathcal{M}$.

\vspace{-0.15cm}
\section{Translating Machines to Probabilistic Rewrite Theories}\label{sec:enc}

Rewriting logic (RL) \cite{meseguer-rltcs-1992} (see a survey in
\cite{meseguer-twenty-2012}) is a general model for concurrency where systems
are declaratively specified using algebraic data types and
(conditional) rewrite rules. This section presents a map $\os \cdot \cs : \mcal \to
\rcal_\mcal$ from a probabilistic Event-B machine to a probabilistic rewrite
theory \cite{agha-pmaude-2006}. The rewrite theory $\rcal_\mcal$ is obtained in
two steps. First, it contains (as a subtheory) a rewrite theory $\rcal$
defining sorts and operations to represent declared types/sets and constants of
any context, as well as the infrastructure needed to encode the variables and
the events of any probabilistic machine. Second,  $\rcal$ is extended with rules
specific for the machine $\mcal$. Computation with $\rcal_\mcal$ is shown
to be free of un-quantified non-determinism (a condition needed for
statistical analysis \cite{agha-pmaude-2006}), as well as to be sound and complete
w.r.t. the probabilistic semantics of $\mcal$.

In the following sections, the background on RL needed to understand the translation is
introduced gradually. In most of the cases, the notation of Maude
\cite{clavel-maudebook-2007}, a high-level language   that supports membership
equational logic and rewriting logic specifications, will be adopted. This has
the immediate effect of producing an executable specification. For the sake of
readability, some details of the specification are omitted. The complete
specification of the theory $\rcal_\mcal$ is available at \cite{tool.website}
as well as a parser automatizing the translation $\os \cdot \cs$.

\subsection{Translating Contexts} \label{sec:trans-ctx}

A \emph{rewrite theory} is a tuple $\rcal = (\Sigma, E \uplus B, R)$. The
static behavior  of the system  is modeled by the order-sorted equational
theory $(\Sigma, E \uplus B)$ and the dynamic behavior by the set of rewrite
rules $R$ (more on this in \S \ref{sec:spec-machine}). The signature $\Sigma$
defines a set of typed operators  used to build the terms of the language
(\ie{} the syntax of the modeled system).   $E$ is a set of (conditional)
equations over $T_\Sigma$ (the set of terms built from $\Sigma$) of the form $t
= t' ~ \mathbf{if} \phi$. The equations specify the algebraic identities that
terms of the language must satisfy (e.g., $x+0 = x$). Moreover, $B$ is a set of
structural axioms (associativity, commutativity, and identity, or combinations of them)
over $T_\Sigma$ for which there is a finitary matching
algorithm.
The equational theory thus defines algebraic data types and
deterministic and finite computations as in a functional programming language.

Let's start defining appropriate sorts (types) and operators to specify an
Event-B contexts $\mathcal{C} = \langle SET_l \cup SET_n, CTE\rangle$.
Natural and integer numbers, as well as Booleans, are mapped
directly to the corresponding sorts in Maude.  Any other constant symbol
inhabiting the deferred sets $SET_l$ and $SET_n$ (e.g., \code{open},
\code{close}) is represented as a string.  The sort \code{EBElt} below defines
the values for basic types:

\begin{maude}
fmod EBELT is                        --- Functional module EBELT (equational theory)
  sort EBElt .                       --- Basic values
  op elt : Int -> EBElt        .     --- Nat and Int
  op elt : Bool -> EBElt       .     --- Boolean 
  op elt : String -> EBElt     .     --- Elements of deferred sets 
  op _<_ : EBElt EBElt -> Bool .     --- Order on EBElt
[...]
--- Total order on EBElt
view EBElt< from STRICT-TOTAL-ORDER to EBELT is sort Elt to EBElt . endv
\end{maude}

\noindent Terms \code{elt("open")} and \code{elt(9)} of sort \code{EBElt}
represent, respectively, the constant symbol \code{open} in the deferred set
\code{SOC} and the number $9$. The equations defining the strict total order
\code{\_<\_} (underscores in Maude denote the position of the parameters in an
operator) are defined as expected for elements of the same type (e.g., 
for two integers $n,m$, 
\code{elt(n) < elt(m) = n < m}).  Although
elements of different types are not supposed to be compared in (well-typed)
machines, the definition of $<$ decrees that \code{elt(s) <
elt(b)} for any string $s$ and Boolean $b$ and \code{elt(b) < elt(n)} for any
number $n$.

\begin{maude}
fmod EBSET is                           --- Sets of EBElt 
  pr SORTABLE-LIST-AND-SET{EBElt<} * (sort Set{EBElt<} to EBSet) .  
  op _.._ : Int Int -> EBSet .          --- Building finite subsets of Int 
  op gen-set : String Nat -> EBSet .    --- Building indexed (string) constants 
  op gen-set : List{String} -> EBSet .  --- Building a EBSet from a list of strings 
[...]
\end{maude}

\noindent Sortable lists allow to uniquely
represent sets as lists, which is key for having quantified non-determinism.
The Maude's theory \code{SORTABLE-LIST-AND-SET} is instantiated 
with the total order \code{EBElt<} defined above and the sort
\code{Set\{EBElt<\}} is renamed to \code{EBSet}. 
The defined operators, and their  equations omitted here, are used to
encode deferred sets by enumeration, cardinality or as integer
intervals: \code{gen-set("open" "close")} reduces to
\code{\{elt("open"), elt("close")\}}; \code{gen-set(s,n)} reduces to
the set 

\noindent\code{\{elt(s1),...,elt(sn)\}} and the term \code{n .. m} reduces
to \code{\{elt(n),...,elt(m)\}}.

Pairs of \code{EBElt}s as well as (sortable) lists and sets of pairs
 are built as follows: 

 \begin{maude} 
fmod EBPAIR is  
  pr EBSET .  
  sort EBPair .   --- Pairs of EBElt 
  op _|->_ : EBElt EBElt -> EBPair . 
  [...]

view EBPair from STRICT-TOTAL-ORDER  to EBPAIR is sort Elt to EBPair  . endv

fmod EBRELATION is --- Relations / sets of pairs
  pr SORTABLE-LIST-AND-SET{EBPair} *  (sort Set{EBPair} to EBRel) .
 [...]
\end{maude}

The module \code{EBRELATION} implements most of the operations on sets
and relations defined in the syntax of Event-B. However, as explained in
Section \ref{sec:spec-machine}, the operations and equations needed to
specify expressions including set comprehensions ($\{x.S \mid P(x)\}$
and $\{x.S \mid F(x)\}$) are generated according to the
Event-B model at hand.

Values for constants and variables  are terms
of sort \code{EBType} built from (possibly singleton) sets
of \code{EBElt}s or \code{EBPair}s:

\begin{maude}
fmod EB-TYPE is              
  pr EBRELATION . 
  sort EBType .
  op val : EBSet -> EBType . --- Sets of basic types
  op val : EBRel -> EBType . --- Sets of pairs 
[...]
\end{maude}

\noindent
Hence, the value of the Event-B constant
\code{FCMD}
is encoded as the term 
\code{val(elt(9))}
while
a function $f$ of type 
\code{1..2 * SOC}
might be assigned the value

\noindent\code{val(  (elt(1) |-> elt("open"), elt(2)|-> elt("close"))  )}.

All Boolean and arithmetic operations, and operations on sets and
relations are lifted to operators of sort
\code{EBType} ad-hoc. 
For instance, \code{val(elt(2))+val(elt(3))} reduces
to \code{val(elt(3)+elt(2))} that, in turn, reduces to
\code{val(elt(3+2))} 
\footnote{Operators of sort \code{EBType} are certainly partial and a term  such as \\
\code{val(elt(2))+val(elt("open"))}
does not have any reduction. It is worth noticing that such terms
cannot appear in the encoding of a well-typed Event-B model.}.  The
theory \code{EB-TYPE} thus defines an encoding $\encexp{\cdot}$ from
Event-B expressions into \code{EBType} expressions where Event-B
(arithmetic, Boolean, relational, and set) operators and values are
mapped into the corresponding terms defined in \code{EB-TYPE}. Most of
the cases in the definition of $\encexp{\cdot}$ are immediate. The
interesting cases will be introduced gradually next.

An Event-B context is specified as a mapping from the identifiers of
the deferred sets to \code{EBSet}s and context's constants to their
values in the sort \code{EBType}.  This is the purpose of the
following theory:

\begin{maude}
mod EBCONTEXT is 
  pr CONFIGURATION . pr EB-TYPE  . pr MAP{Qid, EBSet} . pr MAP{Qid, EBType} .
  subsort Qid < Oid .    --- Names for contexts
  op Context  : -> Cid . --- Class for contexts
  op sets :_      : Map{Qid , EBSet}  -> Attribute .        --- Context's user-defined sets
  op constants :_ : Map{Qid , EBType} -> Attribute .        --- Context's constants
  op init-context : Qid -> Object .                         --- Building a context
  var Q : Qid .
  eq init-context(Q) = < Q : Context | sets : init-sets, constants : init-constants  > .
  --- Operators to be instantiated by the implementation of the context
  op init-sets : ->  Map{Qid , EBSet} .      --- Building deferred sets
  op init-constants : ->  Map{Qid , EBSet} . --- Initializing constants 
endm
\end{maude}

It is customary in Maude to represent complex systems by using an object
oriented notation (theory \code{CONFIGURATION}). An object $O$ of a class $C$
is represented by a record-like structure of the form $\langle
O:C~|~a_1:v_1,\cdots,a_n:v_n\rangle$, where $a_i$ are attribute identifiers and
$v_i$ are terms that represent the current values of the attributes. The state
of the system is then a multiset of objects called \emph{configuration}. 
In the above theory, quoted 
identifiers (e.g., \code{'GEAR\_CTX}) are used as object identifiers
(\code{subsort Qid < Oid}). Moreover,  a new class identifier (\code{Cid}) is
defined along with the needed attributes. The term 
\code{init-context(Q)} encodes a context with identifier \code{Q}
and attributes \code{sets} and \code{constants}. 
Logical variables as \code{Q} are implicitly  universally
quantified in equations and rules. 

The theory \code{EBCONTEXT} is common to any Event-B model and each particular
context needs to extend it  with equations populating the definition of the
deferred sets and also initializing the constants. For instance,  in the running example, the
following equations need to be added: 

\begin{maude}
  eq init-sets =       ( 'SUD |-> gen-set("up" "down"),     
                         'SOC |-> gen-set("open" "close")) .
                         'SER |-> gen-set("extended" "retracted"), 
  eq init-constants =  ( 'FCMD |-> val(elt(9)) ) .
\end{maude}

\begin{definition}[Encoding of Contexts]\label{def:enc-ctx}
Let $\mathcal{C}=\langle SET_l, SET_n, CTE \rangle$ be an Event-B
context.
The theory $\ctxTheory$ results from extending
\code{EBCONTEXT} with the following equations: 

\noindent\resizebox{.9\textwidth}{!}{
$
\begin{array}{lll}
    \mbox{\it init-sets} &=& \{ id_S \mapsto \mbox{{\it gen-set}}( (id_1,...,id_m) ) \mid id_S:\{id_1,...,id_m\}\in SET_l  \}
    \cup \{id_S \mapsto \mbox{\it gen-set}(id_S,k) \mid id_S:k \in SET_n\}
 \\
\mbox{\it init-constants} &=& \{ c \mapsto \encexp{e} \mid (c~:~type~:=~e) \in CTE\} 
\end{array}
$
}
\end{definition}

\subsection{Encoding Machines} \label{sec:spec-machine} 

The specification of a machine $\mcal$ in the theory
$\rcal_{\mcal}$ requires some extra sorts and equations, but also, 
(probabilistic)  
rewrite rules. Before continuing with the encoding, the meaning of 
the set of rewrite rules $R$ in  $(\Sigma, E \uplus B, R)$ is explained.

\noindent\textbf{Probabilistic rewrite rules.} A rewrite rule
\(\crl{l(\olist{x})}{r(\olist{x})}{\phi(\olist{x})}\) specifies a pattern
$l(\olist{x})$ that can match some fragment of the system's state $t$ if there
is a substitution $\theta$ for the variables $\olist{x}$ that makes
$\theta(l(\olist{x}))$ equal (modulo the set of structural axioms $B$) to that
state fragment, changing it to the term $\theta(r(\olist{x}))$ in a local
transition if the condition $\theta(\phi(\olist{x}))$ is true.  In a
probabilistic rewrite theory~\cite{agha-pmaude-2006}, rewrite rules can have
the more general form \(
\pcrl{l(\olist{x})}{r(\olist{x},\olist{y})}{\phi(\olist{x})}{\olist{y} {\ :=\ }
\pi(\olist{x})}, \) where some new variables $\olist{y}$ are present in the
pattern $r$.  Due to the new variables $\olist{y}$, the next state specified by
such a rule is not uniquely determined: it depends on the choice of an
additional substitution $\rho$ for the variables $\olist{y}$.  In this case,
the choice of $\rho$ is made according to the family of probability functions
$\pi_\theta$: one for each matching substitution $\theta$ of the variables
$\olist{x}$. 

Probabilistic rewrite theories can be simulated directly in Maude by sampling,
from the corresponding probabilistic functions,  the values for the variables
$\olist{y}$ appearing on the RHS. Moreover, using a Monte-Carlo simulation and
the query language QuaTEx (Quantitative Temporal Expressions), it is possible
to analyze quantitative properties of the system by statistical model checking
(more details on \S \ref{sec:cases}).

%
%

\noindent\textbf{Machine's variables and state.} 
The specification of a machine starts with a theory that
extends \code{EBCONTEXT} with a mapping from the identifiers of the
variables to their values:

\begin{maude}
mod EBMACHINE is
  inc EBCONTEXT  .                                    --- Context specification. 
  op Machine  : -> Cid .                              --- Class for machines
  op variables :_ : Map{Qid , EBType} -> Attribute .  --- Machine's variables 
  op init-variables : ->  Map{Qid , EBType} .         --- Instantiated by the machine at hand
  op init-machine : Qid Qid -> Object .               --- Initial state of the machine
  vars QM QC : Qid .
  eq init-machine(QC, QM)  =     --- Configuration defining the machine QM and its context QC
                  init-context(QC)                              --- Context 
                  < QM : Machine | variables : init-variables > --- Variables
                  < events : Events | state: init-events >.     --- Events 
[...] --- event specification explained below
\end{maude}

The specification of a machine $\mathcal{M}$ as the rewrite theory
$\mathcal{R}_\mathcal{M}$ is obtained by extending the previous theory
in two ways:
(1) the initialization event of the machine gives rise to
an equation that populates the mapping \code{init-variables}.  And
(2), building from the infrastructure in  \code{EB-TYPE}, different equations and rewrite rules
are added to \code{EBMACHINE} in order to encode the semantics of the
machine's events.

For the running example, (1) amounts to define the equation:

\begin{maude}
    eq init-variables = ('handle |-> val(elt("up")))        , ('door   |-> val(elt("close"))),
                        ('gear   |-> val(elt("retracted"))) , ('cmd    |-> val(elt(0))) .
\end{maude}

More generally, for any machine with initialization
$x_1:=E_1,...,x_n:=E_n$ the following equation needs to be added to
\code{EBMACHINE}: 

\begin{equation}\label{eq-vars}
\mbox{\it init-variables} = x_1 \mapsto \encexp{E_1},...,
x_n \mapsto \encexp{E_n}
\end{equation}

For (2), the theory must generate a purely probabilistic transition system
without un-quantified non-determinism, thus
guaranteeing that execution paths in the system form a measurable set \cite{agha-pmaude-2006}.
Otherwise, it is not possible to perform (sound) statistical analyses (\S
\ref{sec:cases}). Hence, the theory \code{EBMACHINE} is extended with:

\begin{enumerate}[label=(\roman*)]
 \item A deterministic mechanism that, given the current state of the
   machine, determines whether an event is enabled or not.
 \item A probabilistic rule that chooses, according to the weights of
   the enabled events, the next event to be executed. This rule is
   common to any Event-B model and thus defined directly in the theory 
   \code{EBMACHINE}.
 \item For each event, a rule that chooses probabilistically the
   parameters of the event (if any) as well as the values for
   probabilistic assignments.  Then, the state of the machine is
   updated accordingly.  The application of this rule will correspond
   to an observable state transition of the machine $\mathcal{M}$.
\end{enumerate}

\noindent
\textbf{Events' state (i). }Starting with (i), the sort \code{EvState}
(declared in \code{EBMACHINE}) defines four possible states for an
event:
\begin{maude}
 sort EvState .                             --- States of events
 ops blocked unknown execute : -> EvState . 
 op enable : NzNat -> EvState .             --- Enable with a given weight w > 0
\end{maude}

\noindent
and the object \code{events} (see \code{init-machine} in the module
\code{EBMACHINE}) stores a list with the state of each event:

\begin{maude}
   sort Event LEvent .                            --- Events and list of events
   op ev : Qid EvState -> Event         .         --- ID of the event and its state
   op state:_ : LEvent -> Attribute     .         --- Attribute for objects of class Events   
   op init-events : -> LEvent .                   --- To be instantiated by the machine at hand
\end{maude}

Initially, all the events are in state \emph{unknown}. Hence, for any
machine with events $e_1,...,e_n$, \code{EBMACHINE} must be extended
with the following equation

\begin{equation}\label{eq:events}
\mbox{\it init-events} =  ev(e_1, unknown) ~\cdots~ ev(e_k, unknown) 
\end{equation}

For each constant and variable in the model, a Maude variable of
sort \code{EBType} is added. This facilitates the definition of the
forthcoming equations and rules. In the running example:
\lstinline{vars $FCMD $handle $gear $door $cmd  : EBType .}
Such variables will appear in the objects \code{Context} and
\code{Machine}. This allows for defining $\encexp{x} =\$x$ when the
Event-B variable or constant $x$ appears in the context of an
expression. The prefix ``\$'' is added to avoid clash of names.
Consistently, the following shorthands will be used (\code{CC},
\code{MM}, and \code{EE} in Maude's snippets):

\noindent$
\begin{array}{l}
\termC = \langle \mbox{\it C} : \mbox{\it Context} \mid  sets : \$sets, constants: \mbox{\it init-constants}\rangle
\\
\termM = \langle \mbox{\it M} : \mbox{\it Machine} \mid variables:  \{ x \mapsto \$x\mid  x \in \vec{x}\} \rangle
 \quad  \termE = \langle \mbox{\it E} : \mbox{\it Events} \mid state: \mbox{\it init-events}  \rangle
\end{array}
$
 
\noindent For each event $e$, \code{EBMACHINE} is
extended with an equation that (deterministically) updates the state of  $e$ from
\emph{unknown} to \emph{blocked} or to \emph{enabled}.  More precisely,
consider an event $e \in \mathcal{E}$ with guard $e_G$, weight $e_W$, and
set of parameter declarations $e_{any}$.  Let $W=
\mbox{\it ebtype2nat}(\encexp{e_w})$ and $B=
\mbox{\it ebtype2bool}(\encexp{e_G})$ where 
\code{ebtype2nat(val(elt(n))) = n} (and similarly for \code{ebtype2bool}). The needed equation is the following: 

\begin{equation}\label{eq:act-evt}
\small
\begin{array}{lll}
\termC~\termM~
 \langle E:\mbox{\it Events} \mid  \mbox{\it state:} LE~ ev(e, unknown)~ LE' \rangle = 
\termC~~\termM~~ \langle E:\mbox{\it Events} \mid  \mbox{\it state:} LE~ NSt~ LE' \rangle
\end{array}
\end{equation}
\noindent
where $NSt$ is the expression
 
\noindent$\mbox{\bf if}~ B \wedge W >0 \wedge \bigwedge\limits_{(y:\in e_y) \in
  e_{any}} \mbox{\it not-empty}(\encexp{e_y}) \mbox{\bf ~ then }
ev(e,enable(W)) \mbox{\bf~ else } ev(e, blocked) \mbox{\bf~ fi } $

The constants of the model ($\termC$) and the current values of the variables
($\termM$) are used to evaluate the Boolean expression of the guard ($B$) and
the integer expression of the event's weight ($W$). If $B$ holds, $W>0$, and
the set of possible values for each parameters is not
empty, the event becomes enabled with weight $W$. Otherwise, it is blocked. The
variables $LE$ and $LE'$, of sort \code{LEvent}, allow to apply this equation
at any position of the list of events. 

\noindent
\textbf{Next event (ii).} 
The next step is to add to  \code{EBMACHINE}
the  probabilistic rule 
\begin{equation}\label{eq:choice}
    \pcrl{l(LE)}{r(LE,p)}{\phi(LE)}{p \ := \ \pi(LE)}
\end{equation}

where: $l(LE) = \langle E:\textit{Events}\mid state:~\textit{LE}\rangle $;
$\phi(LE)=true$ iff all event in $LE$ is either enabled or blocked and 
there is at least one enabled event;
$\pi(LE)$ is a uniform distribution on the interval $[0,W)$ where  $W =
\sum\{w_j\mid ev(e_j, enable(w_j) \in LE)\}$; and $r(LE,p) = \langle
E:\textit{Events}\mid state:~ ev(pick(acc(filter(LE)),p), \textit{execute})\rangle $.
Let $LE'=filter(LE)$ be the list of enabled events in $LE$
with weights $w_1,\cdots,w_n$. 
Moreover, let $LE_A=acc(LE')$ be as $LE'$ where the weight of the 
$ith$ event in $LE_A$ is $\sum\limits_{1\leq j\leq i }w_j$. 
Then, $pick(LE_A,p) = e_k$ iff $e_k$ is the first event
in $LE_A$ whose weight is strictly greater than $p$. 
Intuitively, the rule is enabled only when the state
of all the events is different from \emph{unknown} and at least one event is
\emph{enabled} (otherwise, the system is in a deadlock). 
The enabled events are filtered and their weights accumulated ($LE_A$). 
Hence, an enabled event $e$ with weight $w$ is chosen
for execution with probability 
$w/W$, where $W$ is the sum of the weights of the enabled events at the
current state.

Following \cite{agha-pmaude-2006}, the probabilistic rule 
\eqref{eq:choice} can be written in Maude by generating random numbers appropriately: 

\begin{maude}
crl [next-event] :  < events : Events | state: LE > => 
                    < events : Events | ev(pick(LEA, rand(max-value(LEA))), execute) >
if     all-ready(LE)              --- All e in LE is in state either blocked or enabled 
   /\  one-firable(LE)            --- At least one event is enabled 
   /\  LE' := filter(LE)          --- Extract the enabled events
   /\  LEA := accumulate(LE') .   --- Accumulate the weights 
\end{maude}
Here,  \code{max-value} returns the weight of the last element in \code{LEA} (i.e., 
the sum $W$ of the weights of the enabled events). A
random number in the interval $[0,W)$ is generated and the choice of the
next event follows the distribution $\pi(LE)$ as explained above.

\noindent\textbf{Actions in events (iii). }  Step (iii) consists in
defining a rule, for each event $e$, that updates the state of the
machine when such a rule is applied:
\begin{equation}\label{eq:change}
\begin{array}{lll}
\termC~\termM~\langle\mbox{\it events:Events} \mid \mbox{\it state: ev(e, execute)} \rangle & \Rightarrow &
\termC~ \termM'~\termE
\end{array}
\end{equation}
where
$
\termM' = \langle \mbox{\it M} : \mbox{\it Machine} \mid 
\mbox{\it variables:} \{ x \mapsto \$x
 \mid  {x \in \vec{x_u}} \}  \cup
 \{ x \mapsto \encexp{e_x}\mid  {x := e_x \in e_{act}}\}
\rangle
$
and $\vec{x_u}$ is the set of variables that do not appear in the
LHS of the set of actions/assignments $e_{act}$ of the event.
This rule can be fired only if the event has been chosen for
execution.  The expression $\encexp{e_x}$, updating the state of the
variable $x$, is built as before but new cases need to be considered:
\begin{itemize}
\item $\encexp{y}= \mbox{\it choice}(\mbox{\it makeList}(\encexp{e_y}))$ if $y:\in e_y$ is a parameter of the model; 
\item $\encexp{\{\{ E \}\}} = \mbox{\it choice}(\mbox{\it makeList}(\encexp{E}))$; and
\item $\encexp{\{\{ E_1@p_1,...,E_n@p_n\}\}} = \mbox{\it choice}(
  \mbox{\it accumulate}(\encexp{E_1} @ p_1,...,\encexp{E_n} @ p_n))$.
\end{itemize}

The set of values a parameter can take and the (set) expression $E$ in
the probabilistic assignment $x :=  \{ E \} $ are converted into
lists of \code{EBType}s. Since the elements of $\code{EBSet}$ and
\code{EBRel} are \emph{sortable}, these lists are always generated in
the same order. This guarantees that there is no un-quantified
non-determinism due to the way the elements of the set are arranged.
Similar to the definition of the rule \code{[next-event]}, 
random numbers are used to realize the needed probabilistic 
rewrite rule. More precisely,
the function $choice(L)$ generates a random number $0 \leq i <
size(L)$ and returns the $i$th element of the list $L$, thus following
a uniform probabilistic distribution. The $i$th element of
$L'=\mbox{\it accumulate}(\{E_1@p_1,...,E_n@p_n\})$) is $E_i @ \sum_{0\leq
  j \leq p_j }$.  In this case, the function \emph{choice} generates a
random number $0\leq R < 1$ and returns the first element $E_k @ p_k$
in $L'$ where $R < p_k$.
For illustration,  the tool generates the following rule for the event 
$pcmd$:
\begin{maude}
rl [pcmd] :  CC MM  < events  : Events  | state: ( ev('pcmd, execute) ) > =>
             CC < $MNAME : Machine | variables: 
                      ('handle |-> choice( makeList(val(elt("up")) , val(elt("down")))), 
                       'cmd |-> ($cmd) + (val(elt(1))), 'door |-> $door ...) >  EE .
\end{maude}

\noindent\textbf{Handling set comprehension.}  Sets in parameters and
RHS in assignments can be built using set comprehension.  Such
expressions must be evaluated in both, the equation determining the
state of the event (checking whether the set of values for a parameter
is empty or not) and the rule specifying the state transition.  The
expression $\{x. S \mid P(x)\}$ (resp., $\{x. S \mid F(x)\}$) can be
thought of as the  higher order function \emph{filter} (resp.,
\emph{map}) in functional languages.  Even though it is possible to
use the reflective capabilities of rewriting logic and meta-programming
in Maude to encode higher-order functions
\cite{DBLP:journals/entcs/ClavelDM00}, a
simpler path is followed here. For each expression of the form $\{x. S
\mid P(x)\}$, the tool generates a new operator and two equations of
the following form:
 
\begin{maude}
op $filter-id : EBType Configuration -> EBType .
eq $filter-id(val(empty), C) = val(empty) .   --- Base case
--- Recursive case
eq $filter-id(val((E,S)), (< $CNAME : Context | sets: ($sets),constants: ('c1 |-> $c1,...)>
                           < $MNAME : Machine | variables: ('v1 |-> $v1, ...) > )) =
   $filter-id(val( S ),   (< $CNAME : Context |  ... > < $MNAME : Machine | ... >)) union 
                            if ebtype2bool(exp-P(val(E))) then val(E) else val(empty) fi . 
\end{maude}

\noindent
The first parameter corresponds to the (encoding of the) set of values
$S$ and the second is the configuration giving meaning to the constants
and variables of the model. The term \emph{exp-P} is the \code{EBType}
expression encoding the predicate $P$ and the element $E$ is discarded
(\code{union val(empty)})
if such an expression evaluates to \code{val(elt(false))}.  Hence,
$\encexp{\{x.S \mid P(x)\}}=\mbox{\it filter-id}({\encexp{S}, \termC ~
  \termM})$. A similar strategy is used to encode the set $\{x.S \mid
F(x)\}$. 

Summing up, Event-B models are represented as a term of
sort \code{Configuration} with three objects: $ \langle
C:Context ~|~ CAts \rangle ~ \langle M:Machine ~|~ MAts
\rangle ~ \langle E : Events ~|~ EAts \rangle $ where
\emph{CAts}, \emph{MAts}, and \emph{EAts} are the attributes explained
above.  Moreover, Maude's variables appear in $MAts$
and \emph{CAts} in all the LHS of equations and rules, thus
making available the values of the model to encode
arbitrary Event-B expressions as \code{EBType}
expressions.

\begin{definition}[Encoding]\label{def-enc-machine}
Let $\mathcal{M}=  \langle\mathcal{C}, \vec{x}, I, \mathcal{E}, init \rangle$.
The rewrite theory 
$\mathcal{R}_{\mathcal{M}}$
specifying the behavior of 
$\mathcal{M}$ is
obtained by extending \code{EBCONTEXT}
 as in Definition \ref{def:enc-ctx} and extending the theory 
\code{EBMACHINE} with:  the 
 equations
 \eqref{eq-vars},
 \eqref{eq:events}, 
 \eqref{eq:act-evt}; the rule 
 \eqref{eq:change} (for each event);
 and the operators and equations
 defining filter- and map-like expressions. 
Given a state $\stateSym$
of the machine $\mathcal{M}$,
$\os\mathcal{M}\cs_\stateSym$ denotes the term  $\termC_\stateSym~~\termM_\stateSym~~\termE$
where 
$\termC_\stateSym$
is as $\termC$
 but  each variable $\$c_i$ is replaced with
$\encexp{\stateSym(c_i)}$ and 
$\termM_\stateSym$ replaces each variable  $\$x_i$ with
$\encexp{\stateSym(x_i)}$.
\end{definition}

Assuming that the specification of the arithmetic, Boolean, and
relational operators in $\mathcal{R}_{\mathcal{M}}$ is adequate with
respect to the corresponding semantics for the Event-B operators, it
is possible to show that the specification is correct in the following
sense.

\begin{theorem}[adequacy]
\label{th-adq}
Let $\mathcal{M}$ be an Event-B machine and
$\mathcal{R}_{\mathcal{M}}$ be as in Definition \ref{def-enc-machine}.
Hence, $\stateSym\rede{p,e} \stateSym'$ iff
$\os\mathcal{M}\cs_\stateSym \redi_p \os\mathcal{M}\cs_{\stateSym'}$,
where $\redi_p$ denotes one-step rewriting in
$\mathcal{R}_{\mathcal{M}}$ with probability $p$.
\end{theorem}

\begin{proof}(sketch).
  In what follows, $\rightsquigarrow$ denotes equational reduction in
  $\mathcal{R}_{\mathcal{M}}$.
Note that $p$ must be of the form $p_w\times q$ where $p_w$ is the
  first term (weight) in the transition probability function (\S
  \ref{sec:prob-sem}) and $q$ is the rest of the expression (resulting
  probability for parameters and assignments).
  $\os\mathcal{M}\cs_\stateSym$ necessarily exhibits the following
  reductions: $ \os\mathcal{M}\cs_\stateSym =
  \termC_\stateSym~\termM_\stateSym~\termE \rightsquigarrow^*
  \termC_\stateSym~\termM_\stateSym~\termE' \redi_{p'_w}
  \termC_\stateSym~\termM_\stateSym~\termE'' \redi_{q'}
  \termC_\stateSym~\termM'_\stateSym~\termE $ where: $\termE'$ is the
  (unique) normal form of $\termE$ where the state of the events are
  either $blocked$ or $enabled$; $\termE''$ results from an
  application of the rule {\it next-event}, choosing probabilistically
  one of the enabled events; and $\termM'_\stateSym$ is the new state
  after the application of the rule corresponding to the (unique)
  chosen event.  If the semantics of the operators in Event-B agrees
  with the one defined for terms of sort \code{EBType}, (i.e., the
  Event-B expression $e$ reduces to the value $v$ iff $\encexp{e}
  \rightsquigarrow^* \encexp{v}$) the set of enabled events in
  $\termE'$ must necessarily coincide with those enabled in state
  $\stateSym$ and necessarily $p_w = p'_w$. Under the same assumption,
  the set of possible values for the parameters are the same in
  $\stateSym$ and $\termM_\stateSym$. Therefore, $q = q'$ and the term
  $\termC_\stateSym~\termM'_\stateSym~\termE $ necessarily corresponds
  to $\os\mathcal{M}\cs_{\stateSym'}$.
\end{proof}

\section{Case Studies and Experiments}\label{sec:cases}
Given a model $\mathcal{M}$, the theory $\mathcal{R}_{\mathcal{M}}$
generated by the tool \cite{tool.website} can be used to perform
statistical analysis.  The resulting Maude's module includes a
template to ease the specification of properties as formulas in the
Quantitative Temporal Expressions language
(QuaTEx) \cite{agha-pmaude-2006}. 
QuaTEx expressions can query the
expected value of any expression, thus generalizing probabilistic computation
tree logic (PCTL) \cite{DBLP:journals/fac/HanssonJ94}.

QuaTEx formulas can be evaluated by
performing a Monte-Carlo simulation with the aid of the PVeStA
tool~\cite{alturki-pvesta-2011}. Hence, the expected value of a given
QuaTEx expression can be obtained within a confidence interval
according to a parameter $\alpha$.
PVeStA defines the operator \code{val: Nat Configuration -> Float}
that identifies the QuaTEx expression (first parameter) and, given a
configuration, returns the value of the expression as a float.  The
parser of the tool accepts a section \code{PROPERTIES}, at the end of the machine
definition, including Event-B expressions that are later translated
and used in equations giving meaning to \code{val}.  In the running
example, the expression \code{door=open} is translated to the
following equation that reduces to $1.0$ when the doors are open:

\begin{maude} 
eq val(1, Conf < $MNAME : Machine | variables: ... ) = toFloat((($door) =b (val(elt("open"))))) .  
\end{maude}
\vspace{-0.2cm}
The simulation shows that the expected value of such expression is 0.0
(the doors are always closed after finishing the maneuver).  Also, it
is possible to estimate that the probability of ending the sequence of
actions with the gear retracted (property \code{gear=retracted}) is
$0.49\pm 0.01$.  This is explained by the fact that the event $pcmd$
may change the value of the handle to $up$ and $down$ with equal
probability.

As a more compelling example, consider the model of the P2P protocol
in~\cite{DBLP:journals/sosym/AouadhiDL19}.  A file partitioned into
$K$ blocks is to be downloaded by $N$ clients; a client can only
download a block at a time.  As a model of failure, some blocks may be
lost and then retransmitted.  The context of the model includes two
constants $N$ and $K$ of type $Nat$.  A deferred set $STATE = \{ emp,
ok, downloading \}$ defines the state of the blocks. The file is
modeled as a variable of type \code{POW(Nat * State)} and initialized
with the expression \code{(0 .. (N * K - 1)) * \{emp\}}.
The pair $i \mapsto emp$ means that the block $i / N$ has not been
downloaded by the client $i \mbox{ mod } N$. Below the events for
receiving and sending blocks:
 
\begin{EventB}
EVENT sent 
  WEIGHT N * K  - card(file  $\ranres$ {downloading})
  ANY block $:\in${x . dom(file $\ranres$ {emp}) | ((x mod N) $\notin$ {y . dom(file $\ranres$ {downloading}) | y mod N})}
  WHERE True 
  THEN file := file $\ovl$ { block |-> downloading }      
       n    := n + 1 
END
EVENT receive 
  WEIGHT 1 + card( file $\ranres$ {ok})
  ANY block $:\in$ dom(file $\ranres$ { downloading }) 
  WHERE True 
  THEN file := file $\ovl$ { block |-> ok } 
END
\end{EventB}

The event $sent$ selects a block $x$ s.t. $file(x)=emp$ (range restriction $\ranres$)
and whose
client ($x \mbox{ mod } N$) is not in the set of clients
that are currently downloading blocks.
Note the nested set comprehension expression.  In that case, the
variable $file$ is updated ($\ovl$) by changing the state of the
block to $downloading$ and incrementing the counter $n$. The event
$receive$ selects one of the blocks in state $downloading$ 
and updates it to the state $ok$, signaling that it was successfully
downloaded.  The probability of sending blocks decrements according to
the number (\code{card}) of blocks being downloaded.

\noindent\begin{wrapfigure}{l}{0.47\textwidth}
\begin{minipage}{.47\textwidth}
\vspace{-0.2cm}
\begin{EventB}
EVENT fail 
 WEIGHT N * K - card( file $\ranres$ {ok})
 ANY block $:\in$ dom(file $\ranres$ { downloading }) 
 WHERE True
 THEN file := {file @ 0.6 ,
              (file$\ovl${ block |-> emp })@ 0.4} 
END
\end{EventB}
\end{minipage}
\vspace{-0.8cm}
\end{wrapfigure}
The event \emph{fail} 
leaves the file unchanged with probability 0.6 and,
with probability 0.4, it changes the state of one
$downloading$ block to $emp$. In the second case, the
block needs to be retransmitted. 
For instance, if $N=16$ and $K=30$, 
PVeStA reports that the expected value for $n$  is $1554.56$ (i.e., each block
is transmitted, in average, 3.24 times). Other experiments
for different values of $N$ and $K$ as well as more details 
about the simulations can be found at  \cite{tool.website}.
The site of the tool contains also  other case  studies including  the probabilistic model for the emergency
brake system described in \cite{aouadhi:hal-01316599} and 
the bounded re-transmission protocol modeled in Event-B in
 \cite[Chapter 6]{DBLP:books/daglib/0024570}.


 \vspace{-0.1cm}
\section{Concluding Remarks}\label{sec:conc}
Combining  formalisms in the context
of the B-method has been explored in different directions, thus leading to 
more robust tools for system modeling and verification. 
The authors of \cite{afendi-eventbhybrid-2020} propose a
correct-by-construction approach for hybrid systems in Event-B.  The
main idea is to move from an event-triggered to a time-triggered
approach because of real-life scenarios. However, models become more
difficult to verify. For this latter purpose, the authors use a
dynamic logic for refinement relations on hybrid systems to prove that
time-triggered models are refinements of event-triggered models. There
is a wide-range effort by the ANR agency and its partners in the EBRP
project~\cite{ebrp} to enhance Event-B and the corresponding Rodin \cite{DBLP:books/daglib/0024570} 
toolset by defining extension mechanisms. Their main goal is to allow
Event-B models to import and use externally defined domain theories
via theory constructs already available in Event-B and implemented in
Rodin as a plug-in. The work in 
\cite{tarasyuk-eventbhybrid-2015} proposes an extension of Event-B
to enable stochastic reasoning about dependability-related
non-functional properties of cyclic systems.  Such an extension
integrates reasoning about functional correctness and stochastic
modeling of non-functional characteristics. Recently, the authors in 
\cite{DBLP:journals/corr/abs-2108-07878}
 have proposed B Maude, a prototype executable
environment for the Abstract Machine Notation (AMN) implemented in the
Maude language. It endows the B method with execution by rewriting,
symbolic search with narrowing, and Linear Temporal Logic model
checking of AMN descriptions.

Following the lines of the aforementioned works, this paper couple Event-B
with further tools and reasoning techniques by proposing  a rewriting logic
semantics for a probabilistic extension of it. The translation from an Event-B
model to a probabilistic rewrite theory has been fully automated, and it was
shown to be sound and complete w.r.t. the semantics of the model. The
translation supports a wide spectrum of operators and constructs usually
present in Event-B specifications and all sources of probabilistic behavior
 as proposed in \cite{DBLP:journals/sosym/AouadhiDL19}. The
resulting probabilistic rewrite theory can be executed in
Maude~\cite{agha-pmaude-2006} and statistically model checked with the PVeStA
tool~\cite{alturki-pvesta-2011}. A case study has been presented to illustrate
the encoding and how the statistical analysis enabled by the translation can
complement the inference-based approach in Event-B.

Future work stems from different needs. It is worth investigating how to
incorporate other distribution functions to govern probabilistic choices
(including concurrency/interleaving in events) in Event-B models. As shown in \S
\ref{sec:spec-machine}, probabilistic rewrite theories can incorporate
arbitrary probabilistic distribution functions. Hence, the semantics proposed
here can be used to experiment with different alternatives and propose suitable
extensions for Event-B. Support for new operators and constructs (e.g.,
arbitrary types in Cartesian products) could be included in a new version of
the translation. Furthermore, a plug-in to integrate the tool proposed here to
Rodin is currently under development. 
The framework presented here  supports also non-probabilistic Event-B models
(see the bounded re-transmission protocol \cite{DBLP:books/daglib/0024570} in the tool's site
whose only source of probabilistic behavior is the weight on events).
It will be interesting to investigate the use of  symbolic techniques, such 
as rewriting modulo SMT,  for synthesis of parameters
in Event-B specifications (e.g., finding the range of values for constants
that makes the invariants true). 


\end{document}